\documentclass[reprint, showpacs, amsmath,amssymb, aps, pra, nofootinbib]{revtex4-1}

\usepackage{graphicx}				

\usepackage{array}								
\usepackage{amssymb}
\usepackage{amsmath}
\usepackage{mathtools}
\usepackage{braket}
\usepackage{amsthm}
\usepackage{dsfont}							
\usepackage{verbatim}						
\usepackage{MnSymbol}						
\usepackage{mathrsfs}							

\newtheorem{thm}{Theorem}

\newtheorem{lemma}{Lemma}[section]
\newtheorem{corollary}{Corollary}

  {\renewcommand{\thecorollary}{\ref{#1}$'$}%
   \addtocounter{corollary}{-1}%
   \begin{corollary}}
  {\end{corollary}}

\DeclareMathOperator{\tr}{Tr}
\DeclareMathOperator{\rank}{rank}

\DeclarePairedDelimiter\floor{\lfloor}{\rfloor}

\begin{document}


\title{Optical quantum communication complexity in the simultaneous-message-passing model}


\author{Ashutosh Marwah}
\email[]{ashutosh.marwah@outlook.com}
\affiliation{Institute for Quantum Computing and Department of Physics and Astronomy, University of Waterloo \\ 
D\'epartement d'informatique et de recherche op\'erationnelle, Universit\'e de Montr\'eal}

\author{Dave Touchette}
\email[]{touchette.dave@gmail.com}
\affiliation{Institute for Quantum Computing and Department of Combinatorics and Optimization, University of Waterloo, Perimeter Institute for Theoretical Physics, Departement d'informatique and Institut Quantique, Universit\'e de Sherbrooke}


\date{\today}

\begin{abstract}
The communication cost of a classical protocol is typically measured in terms of the number of bits communicated for this determines the time required for communication during the protocol. Similarly, for quantum communication protocols, which use finite-dimensional quantum states, the communication cost is measured in terms of the number of qubits communicated. However, in quantum physics, one can also use infinite-dimensional states, like optical quantum states, for communication protocols. Communication cost measures based on counting the (equivalent) number of qubits transmitted during communication cannot be directly used to measure the cost of such protocols, which use infinite-dimensional states. Moreover, one cannot infer any physical property of infinite-dimensional protocols using such qubit based communication costs. In this paper, we provide a framework to understand the growth of physical resources in infinite-dimensional protocols. We focus on optical protocols for the sake of concreteness. The time required for communication and the energy expended during communication are identified as the important physical resources of such protocols. In an optical protocol, the time required for communication is determined by the number of time-bin modes that are transmitted from one party to another. The mean photon number of the messages sent determines the energy required during communication in the protocol. We prove a lower bound on the tradeoff between the growth of these two quantities with the growth of the problem size. We call such tradeoff relations \emph{optical quantum communication complexity relations}. 
\end{abstract}

\pacs{}

\maketitle

\section{Introduction}
\begin{sloppypar}
Communication complexity studies the amount of communication required by two parties in order to compute a particular function $f$ on their private inputs. In classical communication complexity, the amount of communication required during a protocol is quantified by the number of bits the two parties communicate. Analogously in quantum communication complexity, the number of qubits communicated is used for this purpose. The field of quantum communication complexity is interesting both because it offers significant advantages compared to the classical setting \cite{Buhrman98, Buhrman01, Yossef04, Gavinsky07} and also because it allows us to understand the fundamental properties of quantum physics \cite{Cleve97, Brassard06, Shi08}. \\

In qubit based quantum communication protocols, a single qubit is viewed as a \emph{unit} of communication. If one were to use particles with $d$-dimensional quantum states or \emph{qudit} for communication during a protocol, the communication complexity would only be linearly scaled by a factor of $(\log(d))^{-1}$ as compared to a protocol using qubits. Thus, it is sufficient to study qubit based protocols in order to understand the communication complexity of qudit based protocols as well. However, quantum mechanics also allows the parties involved in the communication protocol to send individual particles whose state is described by infinite dimensional Hilbert spaces. In fact the light modes used in most optical implementations, which are one of the most common and easiest ways of implementing quantum communication protocols, have an infinite dimensional Hilbert space associated with them. For these protocols, one can no longer directly use the qubit based communication complexity lower bounds. One way to measure the complexity of these protocols would be to estimate the number of qubits that would be required to approximate the infinite dimensional states so that the error in a protocol implemented using these states is negligibly different from the original protocol \cite{Arrazola14, Xu15, Guan16}. Another way is to instead measure the complexity of these protocols using the amount of information transmitted \cite{Arrazola16, Touchette18, Lovitz18}. Alternatively, one can view the infinite dimensional particles themselves as \emph{units} of communication in these protocols and we can count the number of such particles communicated during the protocol to measure the complexity of these protocols. This would provide us greater insight into the physical resources required for communication during such protocols. In this paper, we use this approach to study optical quantum communication protocols.\\

In order to define the complexity of a protocol this way, one has to constrain these particles according to some measure, otherwise such a complexity measure would be trivial as one can always embed an arbitrarily large Hilbert space in the infinite dimensional Hilbert space of a single particle. In particular, one of the parties can encode her input on the infinite orthonormal basis of the Hilbert space of her particle and send it to the other party, who can use it to compute the function. In time-bin encoded optical protocols, the number of modes (denoted by $m$) used in a protocol determines the duration of time required for communication during the protocol. Further, as mentioned above the Hilbert space associated with each optical mode is infinite dimensional. The mean photon number of the optical messages sent during a protocol (denoted by $\mu$) determine the energy required during the protocol. In general both $m$ and $\mu$ would depend on the problem size $n$. In this paper, we study the tradeoff between the growth of these two quantities with the problem size and we call such tradeoff relations \emph{optical quantum communication complexity} relations of $f$. We will mainly restrict our attention to optical protocols in the simultaneous message passing (SMP) model. However, we also describe how the  results presented here may be translated to other communication models. The main result of our paper can be informally stated as follows. 

\begin{thm}[Informal statement]
If $\Pi$ is an optical quantum communication protocol which computes the function $f$ in the SMP model with error at most $1/3$, then the number of modes, $m$, and the maximum mean number of photons the parties may be required to send to the referee, $\mu$, during $\Pi$ satisfy
\begin{align}
	\min\{\mu \log(m), m\log(1+\mu/ \delta)\} = \Omega( Q^{||}_{1/3}(f))
\end{align}
where $\delta>0$ is a constant and $Q^{||}_{1/3}(f)$ is the (qubit based) SMP quantum communication complexity for protocols computing $f$ with at most $1/3$ probability of error. In particular, this implies 
\begin{align}
	\min\{\mu \log(m), m\log(1+\mu/ \delta)\} = \Omega( \log(D(f)))
\end{align}
where $D(f)$ is the classical deterministic communication complexity of $f$.
\label{thm:MainResult0}
\end{thm}

We introduce the concepts and results required from quantum optics and communication complexity in Section \ref{sec:QuantumOptics}. The theorem above is proven in Section \ref{sec:OptCommComp}. A comparison with a classical analogue of the above result is provided in Section \ref{sec:ClComp}.

\end{sloppypar}

\section{Background}

\subsection{Quantum optics}
\label{sec:QuantumOptics}

We only require a few basic concepts from quantum optics for the purpose of this paper. We cover all of these briefly in this section. To begin, the Hilbert space for a single optical mode, $\mathcal{H} $, is a countably infinite dimensional Hilbert space, also called the Fock space. Formally, we identify this Hilbert space with $ \ell_2$, the set of all square summable sequences. If $\hat{n}$ is the photon number operator on this Hilbert space, then we can let $\{ \ket{k} \}_{k=0}^\infty$ be the eigenvectors of $\hat{n}$. These form an orthonormal basis for $\mathcal{H} $ called the Fock basis. A concrete way to view $\mathcal{H} $ is as 
\begin{align*}
	\mathcal{H} = \left\lbrace \sum_{k=0}^{\infty} x_k \ket{k} : \sum_{k=0}^{\infty} |x_k|^2 < \infty \right\rbrace.
\end{align*}
The photon number operator on this space is given by
\begin{align*}
	\hat{n}:= \sum_{k=0}^\infty k \ket{k}\bra{k}. 
\end{align*}

The total number operator on the Hilbert space of $m$-modes $\mathcal{H}^{\otimes m}$ is given by $\hat{N} := \sum_{i=1}^m \hat{n}_i$, where $\hat{n}_i := I \otimes \cdots \otimes \hat{n} \otimes \cdots \otimes I$ (the number operator acting on the $i^\text{th}$ Hilbert space). Since $\hat{N}$ is Hermitian, it can also be associated with a measurement. Using Eigenvalue decomposition, write $\hat{N}$ as $\hat{N}= \sum_{n=0}^\infty n P_n$, where $P_n$ is the projector onto the $n$-photon subspace, i.e., 
\begin{align*}
	P_n = \sum_{(n_1,\cdots, n_m) \in S_n } \ket{n_1, n_2, \cdots, n_m}\bra{n_1, n_2, \cdots, n_m}
\end{align*}
where $S_n := \{ (n_1, n_2, \cdots, n_m) : \sum_{i=1}^m n_i =n \}$. The measurement corresponding to $\hat{N}$ is the measurement $\{ P_n \}_n$. We will also refer to the random variable corresponding to the measurement result in this basis as $\hat{N}$. Thus, the probability of measuring $n$-photons in the state $\rho$ will be denoted by 
\begin{align*}
	\mathbb{P}\text{r}_\rho [\hat{N} = n] = \tr (P_n \rho).
\end{align*}
The mean number of photons of the state is given by
\begin{align*}
	\mathbb{E}_\rho [\hat{N}]  =  \sum_{n=0}^\infty n \mathbb{P}\text{r}_\rho [\hat{N} = n] =\tr (\hat{N}\rho).
\end{align*}
Finally, we note that the Markov inequality for $\hat{N}$ (viewed as a random variable) implies that
\begin{align}
	\mathbb{P}\text{r}_\rho [\hat{N} \geq a ] \leq \frac{\mathbb{E}_\rho[\hat{N}]}{a}. 
	\label{eq:MarkovIneq}
\end{align}

\subsection{Communication complexity}
\label{sec:CommComp}
\begin{sloppypar}
Communication complexity is the study of the number of bits two parties need to communicate in order to be able to compute a function on their inputs. There are different models of communication one can consider to quantify the communication complexity of a function. In this paper, we will mainly deal with the Simultaneous Message Passing (SMP) model. We will, however, use results connecting the complexity of a function in the SMP model with the two-party deterministic communication complexity of a function. In this section, we describe these settings and the results we use in this paper. We point the reader to the books \cite{Kushilevitz96, Rao19} for a more thorough introduction to this subject.\\

We begin with an overview of classical communication complexity. Consider two parties Alice and Bob who wish to collaborate and compute a function $f: \mathcal{X} \times \mathcal{Y} \rightarrow \{ 0,1 \}$ on their inputs $x \in \mathcal{X}$ and $y \in \mathcal{Y}$ \emph{exactly} using a deterministic protocol $\Pi$. The number of bits they need to communicate with each other for this purpose is called the communication cost of the protocol on inputs $x$ and $y$ and is denoted by $\text{cost}_{\Pi}(x,y)$. The communication cost of the protocol is given by
\begin{align*}
	\text{cost}(\Pi) := \max_{x,y} \{ \text{cost}_{\Pi}(x,y) \}.
\end{align*}
Further, we define the deterministic communication complexity of a function $f$ to be the minimum communication cost for computing $f$. That is, 
\begin{align*}
	D(f) := \min_{\Pi} \text{cost} \{ \Pi \}
\end{align*}
where the minimization is over all deterministic protocols $\Pi$ which compute $f$ exactly. The definitions given above can also be extended to randomized protocols which allow for an error $\epsilon \geq 0$ during the protocol. A randomized protocol $\Pi$ is said to compute a function $f$ with error at most $\epsilon$, if for every pair of inputs $x,y$ we have 
\begin{align}
	\mathbb{P}\text{r}_{\Pi} [\Pi(x,y) \neq f(x,y)] \leq \epsilon.
	\label{eq:ErrorProb}
\end{align}
The communication cost of a randomized protocol is once again defined to be the maximum number of bits which Alice and Bob may be required to communicate during the protocol, and the randomized communication complexity of $f$ is defined as the minimum protocol cost for computing $f$. \\

The Simultaneous Message Passing (SMP) model is a more restricted setting in communication complexity. In the SMP model, there are three parties: Alice, Bob and a Referee. Alice and Bob receive inputs $x$ and $y$, which are only visible to them. Further, throughout this paper we consider the model where Alice and Bob have access to private randomness as well. Alice and Bob both send messages to the Referee, so that he is able to compute $f(x,y)$ with high probability. A protocol is said to compute function $f$ with error at most $\epsilon$ if it satisfies the condition in Eq. \ref{eq:ErrorProb} for every input $x,y$. The communication cost of a SMP protocol $\Pi$ is the maximum number of bits Alice and Bob have to send to the Referee for any input and randomness. The SMP communication complexity of computing a function $f$ with error at most $\epsilon$ denoted by $R_{\epsilon}^{||} (f)$ is the minimum communication cost of any SMP protocol which computes $f$ with error at most $\epsilon$. \\

We can define similar settings in the quantum case as well. In this paper, however, we will only consider quantum communication protocols in the SMP model. The setting of the model is the same as the classical model above. However, now Alice and Bob can send quantum states as messages to the Referee. In the model that we consider in this paper, there are no shared resources between any of the parties. Now the communication cost of the protocol is quantified using the maximum number of qubits sent by Alice and Bob during the protocol. Suppose $\Pi$ is a SMP quantum communication protocol, then we define $\text{cost}^Q_{\Pi}(x,y)$ to be the total number of qubits sent by Alice and Bob to the Referee. As before, we define the quantum communication cost of the protocol as 
\begin{align*}
	\text{cost}_Q(\Pi) := \max_{x,y} \{ \text{cost}^Q_{\Pi}(x,y) \}.
\end{align*}
The quantum SMP communication complexity of computing $f$ with error at most $\epsilon$ denoted by $Q_{\epsilon}^{||}(f)$ is 
\begin{align*}
	Q_{\epsilon}^{||}(f) := \min_\Pi \{ \text{cost}_Q(\Pi) \}
\end{align*}
where the minimization takes place over quantum SMP protocols which compute $f$ with error at most $\epsilon$.\\

We are typically interested in the asymptotic growth of communication complexity with the size of the inputs. In order to study this growth, we suppose that the parties wish to compute a family of functions $\{ f_n : n \in \mathbb{N}\}$ where $f_n: \{0,1\}^n \times \{0,1\}^n \rightarrow \{0,1\}$, that is, each function $f_n$ is defined for inputs of length $n$. The reader should think of this family as a generalization of a function defined for a particular input length to strings of all possible lengths. For example, the family of equality functions Eq$_n: \{0,1\}^n \times \{0,1\}^n \rightarrow \{0,1\}$ is defined as 
\begin{align}
	\text{Eq}_n(x,y) = 
	\begin{cases}
		1 & \text{if } x=y \\
		0 & \text{if } x\neq y
	\end{cases}
	\label{eq:EqualityFn}
\end{align}
and can be regarded as a generalization of the Equality function to all possible input lengths. We collectively refer to the family of functions $\{ f_n \}_n$ as $f$. The communication complexity $C$ of the family of functions $f$ is a function of the input size $n$ defined as $(C(f))(n):=C(f_n)$, where the communication complexity $C$ can be chosen to be any of the communication complexities defined for a function above. Further, $f$ is commonly referred to as a function instead of a family of functions and we say that a protocol computes the \emph{function} $f$ for inputs $x,y \in \{0,1\}^n$ to mean that the protocol computes $f_n(x,y)$ on these inputs. \\

We will now state some well known results in communication complexity, which will be used later on in this paper. It should be noted that all these results are about the asymptotic communication complexity of a family of functions. Theorem \ref{th:ConfAmp} shows that up to multiplicative factors the communication complexity of a family of functions is the same for different errors. Theorem \ref{th:SMPLB} lower bounds the randomized SMP communication complexity of a family of functions in terms of its deterministic SMP communication complexity. Theorem \ref{th:QProtLB}, on the other hand, lower bounds the quantum SMP communication complexity in terms of the classical randomized SMP communication complexity.

\begin{thm}[Confidence Amplification; see for example Ref. \cite{Kushilevitz96}] Consider a family of functions $f$ as defined above and any $0 < \epsilon, \delta<1/2$. Then, in the SMP setting, we have that 
	\begin{align*}
		R_\epsilon^{||} (f) &= O(R_{\delta}^{||}(f) \phi(\epsilon, \delta)) \\
		Q_\epsilon^{||} (f) &= O(Q_{1/3}^{||}(f) \phi(\epsilon, \delta))
	\end{align*}
	\label{th:ConfAmp}
\end{thm}
where $\phi(\epsilon, \delta)= O(\log(1/\epsilon)/ ((1/2-\delta)^2 (1-\delta)))$ is a function independent of $n$. In other word this theorem implies that $Q_\epsilon^{||} (f) = \Theta(Q_\delta^{||} (f))$ for all $0 < \epsilon, \delta<1/2$.

\begin{thm}[Babai and Kimmel \cite{Babai97}]
The classical SMP communication complexity of a family of functions $f$ as defined above satisfies
	\begin{align*}
		R_{1/3}^{||}(f) = \Omega(\sqrt{D(f)}).
	\end{align*}
	\label{th:SMPLB}
\end{thm}

\begin{thm}[{see for example Ref. \cite[Section 2]{Buhrman01}}]
For any family of functions $f$, the quantum and classical SMP communication complexities are related as follows
		\begin{align*}
			Q_{\epsilon}^{||}(f) = \Omega (\log(R_{\epsilon}^{||}(f))).
		\end{align*}
		\label{th:QProtLB}
\end{thm}
\end{sloppypar}

\section{Optical quantum communication complexity}
\label{sec:OptCommComp}

In order to implement a quantum communication protocol optically, the quantum messages sent by the parties to each other must be implemented as optical states in a multimode Fock space. The communication cost of such protocols is infinite according to the definitions given in Section \ref{sec:CommComp}, since the dimension of the Hilbert space is itself infinite. In standard communication complexity, the number of bits or qubits communicated during the course of the protocol is related to the time that would be spent communicating during the protocol if the protocol were to be implemented. Thus, the qubit communication complexity serves as a good way to quantify the cost of the protocol. However, for optical protocols this viewpoint is no longer valid as the Hilbert space of a single mode is itself infinite dimensional. For time-bin encoded optical protocols, the number of time-bin modes transmitted determines the duration of the protocol. On the other hand, the mean photon number of the signals is directly proportional to the mean
energy carried by the signals and hence determines the energy required to create the signals. In this section, we will study the tradeoff between the number of modes and the mean number of photons required to run an optical quantum SMP protocol. We are essentially studying the tradeoff relation between the time required for communication during the protocol and the energy required for communication. As stated earlier, we call tradeoff relations between these two quantities for any protocol computing the function $f$ the \emph{optical quantum SMP communication complexity} relation of $f$.\\

\begin{sloppypar}
We will study the tradeoff between the number of modes and the mean photon number for a family of optical SMP protocols computing a function $f: \{ 0,1 \}^n \times \{ 0,1 \}^n \rightarrow \{ 0, 1 \}$ defined for every $n$, for example the Equality function or the Inner Product function (i.e., $f$ is a family of functions as defined in Sec. \ref{sec:CommComp}). Let $\{ \Pi_n \}_{n=1}^\infty$ be a family of SMP protocols, which computes the function $f(x,y)$ with error at most $1/3$. The exact value of error is not relevant, since our bounds use the classical and quantum communication complexity lower bounds, which are equal up to multiplicative factors for fixed error rates (Theorem \ref{th:ConfAmp}). The protocol $\Pi_n$ can be used to compute the function $f(x,y)$ when $x$ and $y$ are $n$-bit strings. We suppose that these protocols are implemented optically. That is, the states sent by Alice and Bob while running $\Pi_n$ are part of a $m(n)$-mode Hilbert space $\mathcal{H}^{\otimes m(n)}$, where $\mathcal{H}$ is the single mode Fock space. Note that the states used depend on the problem parameter $n$, and hence the number of modes is a function of n, $m= m(n)$. We will call the states sent by Alice and Bob on inputs $x$ and $y$ during protocol $\Pi_n$, $\rho^{(n)}_x$ and $\sigma^{(n)}_y$. Further, we define the maximum mean number of photons $\mu(n)$, which Alice or Bob may have to send during $\Pi_n$ as \footnote{We could have also defined $\mu(n)$ as a maximum over the sum of the mean number of photons of the states sent by both Alice and Bob. This definition would be equivalent to the one given in Eq. \ref{eq:muDef} upto constant factors and hence would not lead to any change in the final results of this paper.}
\begin{multline}
	\mu(n) := \max \bigl\lbrace \lbrace \tr(\hat{N} \rho^{(n)}_x) : x \in \{ 0,1 \}^n \rbrace \cup \\ 
	\lbrace \tr(\hat{N} \sigma^{(n)}_y) : y \in \{ 0,1 \}^n \rbrace \bigr\rbrace.
	\label{eq:muDef}
\end{multline}
For notational convenience, we will drop the explicit dependence of $\mu(n)$ and $m(n)$ on $n$ and denote the number of modes and the maximum mean photon number by $m$ and $\mu$. \\

Our strategy will be to use the fact that the maximum mean photon number is $\mu$ to find a projector $P$, which has high overlap with the states $\rho^{(n)}_x$ and $\sigma^{(n)}_y$ used in the protocol. The rank or the dimension of this projector will be shown to depend only on $m$ and $\mu$. We will use it to transform the given protocol into another protocol, where Alice and Bob send the finite dimensional states $P\rho^{(n)}_x P/\tr(P\rho^{(n)}_x)$ and $P\sigma^{(n)}_y P/\tr(P\sigma^{(n)}_y)$ on inputs $x$ and $y$. This protocol would require the communication of only $O(\log(\rank(P)))$ qubits, which has to satisfy the known lower bounds for the SMP communication complexity of $f$.\\

Consider a fixed value of $n$. For any state $\rho^{(n)} \in \lbrace \rho^{(n)}_x : x \in \{ 0,1 \}^n \rbrace \cup \lbrace \sigma^{(n)}_y : y \in \{ 0,1 \}^n \rbrace$ sent by Alice or Bob during $\Pi_n$ and $\delta>0$ (a parameter which will be chosen later), using the Markov inequality (Eq. \ref{eq:MarkovIneq}) we have that 
\begin{align}
	\nonumber & \mathbb{P}\text{r}_{\rho^{(n)}} [\hat{N} \geq \mu/ \delta ] \leq \delta \frac{\mathbb{E}_{\rho^{(n)}}[\hat{N}]}{\mu} \leq \delta \\
	\Rightarrow \ & \mathbb{P}\text{r}_{\rho^{(n)}} [\hat{N} < \mu / \delta] \geq 1- \delta.
	\label{eq:MarkovProj}
\end{align}
Define, $P := \sum_{(n_1, \cdots, n_m) \in S_{< \mu/ \delta}} \ket{n_1, \cdots, n_m}\bra{n_1, \cdots, n_m}$ where $S_{< \mu/\delta} := \{ (n_1, n_2, \cdots, n_m) : {\sum_{i=1}^m n_i < \mu / \delta \}}$. We can rewrite Eq. \ref{eq:MarkovProj} as
\begin{align}
	\tr(P \rho^{(n)} ) \geq 1- \delta.
\end{align}
Now, using Lemma \ref{lem:ProjCloseness}, we have that for every $x, y$
\begin{align*}
	 & \frac{1}{2}\left\Vert \rho_x^{(n)} - \frac{P \rho_x^{(n)} P}{\tr (P \rho_x^{(n)} P)} \right\Vert_1 \leq \sqrt{\delta} \\
	 & \frac{1}{2} \left\Vert \sigma_y^{(n)} - \frac{P \sigma_y^{(n)} P}{\tr (P \sigma_y^{(n)} P)} \right\Vert_1 \leq \sqrt{\delta}.
\end{align*}
Using Lemma \ref{lem:SMPTransform}, we can create a SMP protocol for $f$ with error at most $1/3 + 2\sqrt{\delta}$, which uses the states $\lbrace P\rho^{(n)}_x P/\tr(P \rho^{(n)}_x): x \in \{ 0,1 \}^n \rbrace \cup \lbrace P\sigma^{(n)}_y P/ \tr(P \sigma^{(n)}_y) : y \in \{ 0,1 \}^n \rbrace$. This protocol requires the communication of only $O(\log(\rank(P)))$ qubits. Further, the rank of the projector $P$ can be estimated as follows. \\

The rank of the projector is equal to $\vert S_{< \mu/\delta} \vert $ which is equal to the number of non-negative integer solutions of the equation 
\begin{align*}
	\sum_{i=1}^m n_i < \frac{\mu}{\delta}.
\end{align*}
If we introduce a slack variable $s$, this is equal to the number of non-negative integer solutions of 
\begin{align*}
	\sum_{i=1}^m n_i +s = \floor{\frac{\mu}{\delta}}
\end{align*}
which is equal to  
\begin{align*}
	{a +m}\choose{m}
\end{align*}
for $a= \floor{{\mu}/{\delta}}$ using standard combinatorics. Further, using Lemma \ref{lemm:CnmBound}, this expression can be bounded by $(1+m )^a$ and $(1+a )^m$. Thus, we have that 
\begin{align}
	\log(\rank(P)) \leq \frac{\mu}{\delta}\log(1+m) \label{eq:ProjDim1}\\
	 \log(\rank(P)) \leq m \log (1+\frac{\mu}{\delta}) \label{eq:ProjDim2}
\end{align}

We can choose $\delta = 10^{-4}$, so that the error of the protocol is strictly smaller than $1/2$. Since, quantum communication communication complexity is asymptotically equivalent up to constant factors for all errors less than $1/2$ (Theorem \ref{th:ConfAmp}), this does not affect the asymptotic bounds we are working towards. Thus, moving forward we can ignore the dependence of the upper bound in Eq. \ref{eq:ProjDim1} on $\delta$. Further, if assume $m \geq 2$ (we can always add an extra mode to the messages if necessary), then we can simplify the bound to 
\begin{align}
	\log(\rank(P)) = O(\mu \log (m)).
\end{align}
Our modified protocol, which uses only finite dimensional quantum states, leads us to the following bound which links the growth of optical resources of a SMP protocol of $f$ with its qubit based SMP communication complexity $Q^{||}_{1/3}(f)$.
\begin{align}
	\min\{\mu \log(m), m\log(1+\mu/ \delta)\} = \Omega( Q^{||}_{1/3}(f))
	\label{eq:SMPBound}
\end{align}

Recall that the SMP communication cost for quantum protocols for function $f$ is lower bounded by $\Omega(\log(R^{||}_{1/3}(f)))$, where $R^{||}_{1/3}(f)$ is the classical SMP communication complexity for computing $f$ with error at most $1/3$ (Theorem \ref{th:QProtLB}). Further, the classical SMP communication complexity is lower bounded by $\Omega (\sqrt{D(f)})$, where $D(f)$ is the deterministic communication complexity of $f$ (Theorem \ref{th:SMPLB}). Thus, the number of qubits used by any quantum protocol is lower bounded by $\Omega(\log(D(f)))$. For any family of optical quantum SMP protocols for $f$, the following optical communication complexity relations hold true:
\begin{align}
	&{\mu}\log(m) = \Omega(\log(D(f)))  \label{eq:GenBound1} \\
	& m \log (1+{\mu}/\delta) = \Omega(\log(D(f))). \label{eq:GenBound2}
\end{align}
We state the results developed in this section so far as Theorem \ref{thm:MainResult}, which is also a formal restatement of Theorem \ref{thm:MainResult0}.

\begin{thm}
If $\Pi$ is a family of optical quantum communication complexity protocols which computes ${f:\{0,1\}^n \times \{0,1\}^n \rightarrow \{ 0,1 \}}$ in the SMP model with error at most $1/3$, then the number of modes $m$ (assuming $m \geq 2$ for the protocol) and the maximum mean number of photons $\mu$ of the protocol $\Pi$ as defined above satisfy
\begin{align}
	\min\{\mu \log(m), m\log(1+\mu/ \delta)\} = \Omega( Q^{||}_{1/3}(f))
	\label{eq:MainResult1}
\end{align}
where $\delta = 10^{-4}$ and $Q^{||}_{1/3}(f)$ is the (qubit based) SMP quantum communication complexity for protocols computing $f$ with at most $1/3$ probability of error. In particular, this implies 
\begin{align}
	\min\{\mu \log(m), m\log(1+\mu/ \delta)\} = \Omega( \log(D(f)))
	\label{eq:MainResult2}
\end{align}
where $D(f)$ is the classical deterministic communication complexity of $f$.
\label{thm:MainResult}
\end{thm}

As an example, let us consider optical quantum SMP protocols for the Equality function (defined in Eq. \ref{eq:EqualityFn}). Protocols to compute the equality function are also called fingerprinting protocols. The deterministic communication complexity for the Equality problem is $\Theta(n)$ (see for example \cite[Theorem 1.15]{Rao19}). If the maximum mean number of photons for a family of quantum fingerprinting protocols is constant or bounded, as is the case with Arrazola and L\"utkenhaus' coherent state quantum fingerprinting (QFP) protocol \cite{Arrazola14}, then we have that
\begin{align*}
	\log(m) = \Omega(\log(n)).
\end{align*}
These bounds show that in a \emph{weak} sense the QFP protocol given by Arrazola and L\"utkenhaus is optimal. We use the phrase \emph{weak} because these bounds do not rule out the possibility of a family of optical protocols with constant mean photon number and sublinear growth of $m$ in $n$, as compared to Arrazola and L\"utkenhaus' protocol where $m =\Theta(n)$. \\

We can extend this method to lower bound the growth of the mean photon number and the number of modes in the messages sent during one-way communication protocols. For protocols in this model too, the bound will be similar to the one obtained in Eq. \ref{eq:MainResult1}. One can also use the method described above to derive optical quantum communication complexity relations for interactive two-way communication protocols from the corresponding lower bounds on the qubit based communication complexity\footnote{The average mean number of photons and the number of modes have to be defined appropriately in the case of interactive two-way communication protocols}. However, for these protocols, the number of rounds of the protocol is also a part of the bounds and as a result these bounds are much weaker. 

\end{sloppypar}

\section{Comparison to bounds for classical optical communication complexity relations}
\label{sec:ClComp}

In this section, we will try to compare the optical quantum communication complexity relations with similar relations for classical optical protocols in the SMP model. In quantum optics, the line between "classical" and "quantum" states is extremely blurry. For example, in quantum optics coherent states are usually viewed as classical states of light \cite{Gerry04, Hillery85, Hudson74}, but in a communication complexity setup even such states can provide a tremendous advantage when compared to classical protocols, which use only orthogonal messages. Arrazola and L\"utkenhaus' coherent state quantum fingerprinting protocol provides an example of this advantage. For the sake of simplicity, we will call an optical state, which is diagonal in the Fock basis, a "classical" state in the following\footnote{It should be noted that this description is equivalent to describing "classical" messages using a $m$-tuple $(n_1, n_2, \cdots, n_m)$, where $n_i$ is an integer, which denotes the "power level" of the signal in the $i^{\text{th}}$ mode. The modes can be viewed as time- bin modes, and the power level of the signal $n_i$ as the ratio of the power observed by the detector used during the protocol and its least count. For example, suppose that a protocol uses a detector, which is able to measure the power of a signal in steps of 0.1 W, then if the power measured in the $i^{\text{th}}$ time-bin is 10 W during the protocol, $n_i$ would be 100 $(=10 W/ 0.1 W)$. In this equivalent representation, $\mu$ would represent the maximum average total power of a message that might be sent during the protocol.\label{fn:clModel}}. Now, suppose Alice and Bob are restricted to using these classical states as messages and asked to compute a function $f$ in the SMP model. Once again we denote the number of modes used during the protocol using $m$ and the maximum mean number of photons using $\mu$. Following the arguments of the previous section, we can modify the states $\rho_x$ and $\sigma_y$ used by Alice and Bob, so that the support of the modified states $\rho'_x$ and $\sigma'_y$ lies in the subspace spanned by ${\{ \ket{ n_1, n_2, \cdots, n_m} : {\sum_{i=1}^m n_i < \mu / \delta \}}}$ for some $\delta \in (0,1)$, which will be determined later. We can once again do this in such a way that the additional error is small. This implies that in the modified protocol, on an input $x$, Alice chooses a pure state $\ket{ n_1, n_2, \cdots, n_m}$ such that $\sum_{i=1}^m n_i < \mu / \delta$ with probability $P_x (n_1, n_2, \cdots, n_m)$ to send to the Referee. Similarly, we can think of Bob also randomly choosing such pure states to send to the Referee. As we showed in the previous section, the number of such pure states is 
\begin{align*}
	{a +m}\choose{m}
\end{align*}
for $a= \floor{{\mu}/{\delta}}$. We can transform this classical optical protocol into a standard classical communication complexity protocol, where all the messages are binary, using $\log ({{a +m}\choose{m}})$ bit messages. Now, we can use the classical SMP lower bound (Theorem \ref{th:SMPLB}) on this protocol to get 
\begin{align}
	&\log \left( {{a +m}\choose{m}} \right)  =  \Omega(R_{1/3}^{||}(f) )= \Omega(\sqrt{D(f)}) \label{eq:ClOptbd1} \\
	\Rightarrow &\min\{\mu \log(m), m\log(1+\mu/ \delta)\}  =  \Omega(\sqrt{D(f)}) \label{eq:ClOptbd2}
\end{align}
for say $\delta = 10^{-4}$. We can see that these classical bounds are exponentially stronger than their quantum counterpart (Eq. \ref{eq:MainResult2}). This is simply because we were able to lower bound the classical optical communication complexity using the classical randomized SMP communication complexity instead of the qubit based SMP communication complexity. Further, we note that this classical optical communication complexity relation is tight. If we trivially implement the optimal classical communication complexity protocol for fingerprinting given by Ambainis \cite{Ambainis96}, which has a communication complexity of $\Theta(n)$, in the optical regime using classical optical states (that is, if a party wishes to send the binary message $x_1 x_2 \cdots x_m $, then in the optical protocol they send $\ket{x_1 x_2 \cdots x_m }$), then the number of modes $m=\Theta(\sqrt{n})$ and the maximum mean number of photons $\mu=O(\sqrt{n})$. For these values, using a standard bound for binomial coefficients (See for example Ref. \cite{Cover06}[Example 11.1.3]) we have
\begin{align*}
	&{{a +m}\choose{m}} \leq 2^{(a+m)h(\frac{m}{a+m})} \\
\end{align*}
where $h(.)$ is the binary entropy. This implies that
\begin{align*}
	\log \left( {{a +m}\choose{m}} \right)  \leq O(\sqrt{n})
\end{align*}
which matches the lower bound provided by Eq. \ref{eq:ClOptbd1}. \\

We would like to point out that Fock states and some of their mixtures are considered highly non-classical in quantum optics (for example they have negative Wigner functions). We use them here for our description of classical protocols because they are orthogonal states and from an operational and mathematical point of view the mixtures of orthogonal states behave classically. Moreover, as we point out in Footnote \ref{fn:clModel}, this model can be used to describe a large class of classical models, which are allowed to send signals of unbounded power. It might be more useful to study the optical communication complexity protocols restricted to coherent states or Gaussian states to get a better understanding of the difference between classical and quantum in this regime.  

\section{Conclusion}

In this paper, we adapt the concept of communication complexity to understand the growth of physical resources for optical protocols. We demonstrate simple lower bounds on the growth of the mean number of photons and the number of modes required to implement optical SMP protocols. As motivated in the Introduction, the communication complexity of optical protocols needs to be studied separately as these protocols do not fit the model used by qubit based quantum communication complexity. These relations are a true analogue of classical communication complexity for optical protocols, as we can infer lower bounds on the time required for communication during a protocol from these bounds. Moreover, these relations are important from a practical point of view, since a lionshare of communication protocols are implemented optically \cite{Xu15, Guan16, Horn05, Kumar19}. Optical communication complexity relations are important to understand the limits of the optical implementations of such protocols (also see \cite{Marwah19}). Further work in this area may also help us develop better optical protocols. \\

This paper leaves several questions open for future work. Firstly, it is also not clear at this point if the bounds obtained in this paper are the tightest possible lower bounds for these resources. If indeed these are the tightest bounds it would be interesting to show this using an example. Further, it would also be interesting to see if one can come up with tighter tradeoff bounds for communication protocols which use only coherent states or Gaussian states, since these are the simplest states to experimentally implement. \\

\section*{Acknowledgment}

We would like to thank Norbert L\"utkenhaus for stimulating questions and discussions regarding the obstacles in the implementations of quantum communication protocols. We also thank him for his comments on the manuscript. We would also like to thank Alan Migdall and Joshua Bienfang for sharing their experience on implementations of quantum communication protocols. The work has been performed at the Institute for Quantum Computing, University of Waterloo, which is supported by Industry Canada. The research has been supported by NSERC under the Discovery Program, grant number 341495 and by the ARL CDQI program.

\appendix

\section{Lemmas required for the proof of the Main Result}

\begin{lemma}[Winter's gentle measurement lemma {\cite[Corollary 3.15]{Watrous18}}] 
	Let $\mathcal{H}$ be a Hilbert space, $\rho \in D(\mathcal{H})$ be a density operators and $P \in \text{Pos}(\mathcal{H})$ a positive operator satisfying $P \leq \mathds{1}$ and $\tr (P \rho) >0$. Then, we have 
	\begin{align*}
		F \left(\rho, \frac{\sqrt{P} \rho \sqrt{P}}{\tr(P \rho)} \right) \geq \sqrt{\tr(P \rho)}.
	\end{align*}	 
	\label{lem:GentleMeas}
\end{lemma}

\begin{lemma} Let $P$ be a projector and $\rho \in D(\mathcal{H})$ be a density matrix in the Hilbert space $\mathcal{H}$, such that $\tr (P \rho ) \geq 1 - \delta$ for $\delta \in (0,1)$. Then, we have that 
\begin{align}
	\frac{1}{2}\left\Vert \rho - \frac{P \rho P}{\tr (P \rho P)} \right\Vert_1 \leq \sqrt{\delta}
\end{align}
\label{lem:ProjCloseness}
\end{lemma}
\begin{proof}
	Let us define $\sigma = \frac{{P} \rho {P}}{\tr(P \rho)}$. As a result of Lemma \ref{lem:GentleMeas}, we have 
	\begin{align*}
		F \left(\rho, \sigma \right) \geq \sqrt{\tr(P \rho)} \geq \sqrt{1 - \delta}.
	\end{align*}
	Using the Fuchs-van de Graaf inequality, the trace distance
	\begin{align*}
		\left\Vert \rho - \sigma \right\Vert_1 & \leq 2\sqrt{1- F(\rho, \sigma)^2} \\
		& \leq 2 \sqrt{1- (1-\delta)} \\
		& \leq 2  \sqrt{\delta} 
	\end{align*}
\end{proof}

\begin{lemma}
	Suppose in a quantum simultaneous message passing (SMP) protocol to compute the function $f$ with error at most $\epsilon$, Alice and Bob send the quantum states $\rho_x$ and $\sigma_y$ on inputs $x$ and $y$. If $\rho_x^\prime$ and $\sigma_y^\prime$ are quantum states such that $1/2 \Vert \rho_x - \rho^\prime_x \Vert_1 \leq \delta$ and $1/2 \Vert \sigma_y - \sigma^\prime_y \Vert_1 \leq \delta$ for all $x$ and $y$, then the states used in the actual protocol can be replaced by these to create a SMP protocol with error at most $\epsilon + 2 \delta$.
	\label{lem:SMPTransform}
\end{lemma}
\begin{proof}
Suppose that on inputs $x$ and $y$, Alice and Bob send the quantum state $\rho_x$ and $\sigma_y$ to the referee, who applies $\Phi_{\text{ref}}$ (quantum-classical CPTP map) to the joint state to compute $f(x,y)$. For such a protocol, we have that for every $x, y$
\begin{align*}
	\frac{1}{2}\Vert \Phi_{\text{ref}} (\rho_x \otimes \sigma_y) - \ket{f(x,y)} \bra{f(x,y)} \Vert_1 \leq \epsilon,
\end{align*}
which is equivalent to saying that for inputs $x$ and $y$ the error probability is less than $\epsilon$. Now, if we replace the states used by Alice and Bob by $\rho_x^\prime$ and $\sigma_y^\prime$ such that $1/2 \Vert \rho_x - \rho^\prime_x \Vert_1 \leq \delta$ and $1/2 \Vert \sigma_y - \sigma^\prime_y \Vert_1 \leq \delta$ for all $x$ and $y$, then we have that for every $x$ and $y$
\begin{align*}
	\frac{1}{2}\Vert \Phi_{\text{ref}} & (\rho^\prime_x \otimes \sigma^\prime_y) - \ket{f(x,y)} \bra{f(x,y)} \Vert_1 \\
	& \leq \frac{1}{2}\Vert \Phi_{\text{ref}} (\rho_x \otimes \sigma_y) - \ket{f(x,y)} \bra{f(x,y)} \Vert_1 \\
	& + \frac{1}{2}\Vert \Phi_{\text{ref}} (\rho_x \otimes \sigma_y) - \Phi_{\text{ref}} (\rho^\prime_x \otimes \sigma^\prime_y) \Vert_1 \\
	& \leq \epsilon +  \frac{1}{2}\Vert \Phi_{\text{ref}} (\rho_x \otimes \sigma_y) - \Phi_{\text{ref}} (\rho^\prime_x \otimes \sigma_y) \Vert_1 \\
	& +  \frac{1}{2}\Vert \Phi_{\text{ref}} (\rho^\prime_x \otimes \sigma_y) - \Phi_{\text{ref}} (\rho^\prime_x \otimes \sigma^\prime_y) \Vert_1 \\
	& \leq \epsilon + 2 \delta
\end{align*}
where we have used the fact that for all $\rho \in D(\mathcal{H})$, $\Vert \rho \Vert_1=1$ and for all CPTP maps $\Phi$, $\Vert \Phi \Vert_1 \leq 1$ {\cite[Corollary 3.40]{Watrous18}}.
\end{proof}

\begin{lemma}
For $n, m \in \mathbb{N}$ we have 
	\begin{align*}
		{{n +m}\choose{m}} \leq \min \lbrace (1+m )^n , (1+n )^m  \rbrace.
	\end{align*}
	\label{lemm:CnmBound}
\end{lemma}

\begin{proof}
Observe that
	\begin{align*}
	{n +m}\choose{m} &= \frac{(n+m)!}{m!\ n!} \\
	& = \frac{m+n}{n} \cdot \frac{m+n-1}{n-1} \cdots \frac{m+1}{1} \\
	& \leq (1+m )^n.
\end{align*}
Since, ${n +m}\choose{m}$$=$${n +m}\choose{n}$ we also have
\begin{align*}
	{{n +m}\choose{m}} \leq (1+n )^m. 
\end{align*}
The bound follows by taking the minimum of these two upper bounds. 
\end{proof}

\bibliographystyle{unsrt}
\bibliography{bib}

\end{document}